\documentclass[pra,aps,twocolumn,superscriptaddress,showpacs]{revtex4}

\usepackage{amsmath}
\usepackage[standard]{ntheorem}
\usepackage[colorlinks=true,citecolor=blue,urlcolor=blue]{hyperref}
\begin{document}
\title{Dimension-Independent Bounds for Hardy's Experiment}

\author{Zhen-Peng Xu}
 \affiliation{Theoretical Physics Division, Chern Institute of Mathematics, Nankai University,
 Tianjin 300071, People's Republic of China}

\author{Hong-Yi Su}
\email{hysu@mail.nankai.edu.cn}
 \affiliation{Theoretical Physics Division, Chern Institute of Mathematics, Nankai University,
 Tianjin 300071, People's Republic of China}

\author{Jing-Ling Chen}
\email{chenjl@nankai.edu.cn}
 \affiliation{Theoretical Physics Division, Chern Institute of Mathematics, Nankai University,
 Tianjin 300071, People's Republic of China}
 \affiliation{Centre for Quantum Technologies, National University of Singapore,
 3 Science Drive 2, Singapore 117543}

\begin{abstract}
Hardy's paradox is of fundamental importance in quantum information theory. So far, there have been two types of its extensions into higher dimensions: in the first type the maximum probability of nonlocal events is roughly $9\%$ and remains the same as the dimension changes (dimension-independent), while in the second type the probability becomes larger as the dimension increases, reaching approximately $40\%$ in the infinite limit. Here, we (i) give an alternative proof of the first type, (ii) study the situation in which the maximum probability of nonlocal events can also be dimension-independent in the second type, and (iii) conjecture how the situation could be changed in order that (ii) still holds.
\end{abstract}

\maketitle

\section{Introduction}
Hardy's paradox \cite{Hardy92,Hardy93,Goldstein94,cereceda04,CBTB13} is one of central topics in quantum information science, and has been regarded as ``the simplest form of Bell's theorem." \cite{Mermin95} Conventionally, the Bell nonlocality is detected by the violation of a certain Bell's inequality, showing a contradiction between quantum correlations and local hidden variable (LHV) models. To push the contradiction to extremes, Greenberger, Horne, and Zeilinger (GHZ) proposed a tripartite paradox to prove Bell's theorem without inequalitites, and demonstrated that the probability of nonlocal events is $100\%$. But Hardy's paradox is advantageous over the GHZ paradox, since in the former only two parties involved are enough to rule out LHV models. In this sense, Hardy's paradox is the simplest form of a nonlocality test.

However, the drawback of Hardy's paradox is that the probability of nonlocal events is just $9\%$, far lower than that in the GHZ paradox. Thus it is significant to extend  Hardy's original argument to a form that possesses some higher probability than $9\%$. Efforts include (i) increasing the number of outcomes (i.e., high dimensional) for each party \cite{KC05,SG11,RZS12,JAZ13} and (ii) considering more settings in question \cite{Hardy97,BBDH97}. Collaborating with some other coauthors, in Ref. \cite{JAZ13} we proposed a high-dimensional bipartite Hardy's paradox which is equivalent to a tight Bell's inequality, and numerically showed that the probability of nonlocal events can be over $40\%$.

In literatures there have been two types of bipartite extensions of Hardy's original paradox: (a) proposals in Ref. \cite{KC05,SG11,RZS12} in which the maximum probability of nonlocal events is dimension-independent and has a limit of $9\%$, (b) the paradox we presented in \cite{JAZ13} in which the probability can otherwise go beyond.

In this paper, we consider these two types of Hardy's paradox. For the first type, we give an alternative proof of the dimension-independent maximum probability of nonlocal events. We study the situation in which the probability in the second type could be decreased to that in the first type. With numerical evidence we also conjecture how the situation could be changed in the second type, so that the maximum probability of nonlocal events is still dimension-independent but higher than $9\%$.

\section{Preliminaries}
There is an advantage for bipartite pure state that it can be expressed in the form of matrix, so that matrix theory can be used to deal with problems considered here. Denote $|\psi\rangle = \sum\limits_{i = 1}^m \sum\limits_{j = 1}^n h_{ij} |ij\rangle$ as $|\psi\rangle \rightarrow H = (h_{ij})_{m \times n}$, where $|H| = (\sum_{ij} |h_{ij}|^2)^{\frac{1}{2}} = 1$. Then we have the following two lemmas:

\begin{lemma}
Define matrix $U$ (in the original basis) as a new setting for Alice and $V$ for Bob, then a state $H$ in the basis of the new settings is $U^* H V^\dagger$.
\end{lemma}
\begin{proof}
Denote the original bases of Alice and Bob both as $|0\rangle, |1\rangle$, the new setting as $|0_a\rangle, |1_a\rangle$ and $|0_b\rangle, |1_b\rangle$ for Alice and Bob, respectively. and so
\[
\begin{bmatrix}
|0_a\rangle\\
|1_a\rangle
\end{bmatrix}
=
U
\begin{bmatrix}
|0\rangle\\
|1\rangle
\end{bmatrix},
\begin{bmatrix}
|0_b\rangle\\
|1_b\rangle
\end{bmatrix}
=
V
\begin{bmatrix}
|0\rangle\\
|1\rangle
\end{bmatrix}.
\]
In this way, the state $H$ is expressed as
\[
[|0\rangle, |1\rangle] H \begin{bmatrix}
|0\rangle\\
|1\rangle
\end{bmatrix} =
[|0_a\rangle, |1_a\rangle] (U^\dagger)^T  H V^\dagger\begin{bmatrix}
|0_b\rangle\\
|1_b\rangle
\end{bmatrix},
\]
Thus, in the basis of new setting, the state is expressed as $U^* H V^\dagger$.
\end{proof}

\begin{lemma}
Given an $n$-dimension vector $\vec{x}$ and an $m \times n$-dimension matrix $A$, with $\vec{a}_i$ as the $i$-th row of $A$, then $|A\vec{x}| \leq |A||\vec{x}|$ must hold, where `$=$' is achieved only when $\vec{a}_i^\dagger$ is parallel to $\vec{x}$ for each $i$.
\end{lemma}
\begin{proof}
\begin{eqnarray}
|A\vec{x}|^2 &=& |(\vec{a}_1\vec{x}, \cdots, \vec{a}_n\vec{x})|^2 = \sum_{i=1}^n |\vec{a}_i\vec{x}|^2 \nonumber\\
&\leq& \sum_{i=1}^n |\vec{a}_i|^2 |\vec{x}|^2 = |A|^2 |\vec{x}|^2,
\end{eqnarray}
where `$=$' is achieved only when $|\vec{a}_i\vec{x}| = |\vec{a}_i||\vec{x}|$ for each $i$, i.e., $\vec{a}_i^\dagger$ is parallel to $\vec{x}$.
\end{proof}

For convenience, we will take the symbol conventions as following:
\begin{enumerate}
\item $|\psi\rangle \to H =(h_{ij})$ in the basis of $A_1, B_1$;
\item the setting of $A_2$ in the basis of $A_1$ as $U^\ast$ while the one of $B_2$  in the basis of $B_1$ as $V^\dagger$;
\item $H' =(h'_{ij}) = UH$, $H'' =(h''_{ij}) = HV$ and $Q = UHV$, so $|\psi\rangle \to H'$ in the basis of $A_2, B_1$, $|\psi\rangle \to H''$ in the basis of $A_1, B_2$, and $|\psi\rangle \to Q$ in the basis of $A_2, B_2$;
\item $\vec{h}_{12} = \left(
             \begin{array}{ccc}
               h_{12} & \cdots & h_{1n} \\
             \end{array}
           \right),
\vec{h}_{21} = \left(
             \begin{array}{ccc}
               h_{21} & \cdots & h_{n1} \\
             \end{array}
           \right)^T,$
$
\vec{t} = \left(
             \begin{array}{ccc}
               h_{22} & \cdots & h_{2n} \\
             \end{array}
           \right),
H_{22} = \left(
               \begin{array}{ccc}
                 h_{22} & \cdots & h_{2n} \\
                 \vdots & \ddots & \vdots \\
                 h_{n2} & \cdots & h_{nn} \\
               \end{array}
               \right).
$
\item $P_I(i,j) = P(A_2 = i, B_2 = j)$ in the first type of extension and $P_{II}(i,j) = P(A_2 = i, B_2 = j)$ in the second type of extension.
\end{enumerate}

\section{The First type of extensions of Hardy's Paradox}
The first type \cite{KC05,SG11,RZS12} of extensions of the original Hardy's paradox into high-dimension reads
\begin{eqnarray}\label{kcond}
P(A_1 = 1, B_1 = 1) = 0,\nonumber\\
P(A_1 \neq 1, B_2 = 1) = 0, \nonumber\\
P(A_2 = 1, B_1 \neq 1) = 0, \nonumber\\
P_I = P(A_2 = 1, B_2 = 1) > 0,
\end{eqnarray}
which, in the notations in Sec. II, implies that $h_{11} = h'_{1i} = h''_{i1} = 0$, with $i = 2, \cdots, n$.

It has been conjectured in ref. \cite{SG11} and proved in ref. \cite{RZS12} that
\begin{theorem}
$
\max P_I =\max |q_{11}|^2 = \frac{5\sqrt{5} - 11}{2}.
$
\end{theorem}

Here we present another brief and direct proof.
\begin{proof}		
Without lose of generality, we assume $H_{22}$ is invertible here and $P_I$ is continuous.\\
From condition \eqref{kcond} we know:
\begin{flalign*}
\vec{h}_{21} v_{11} + H_{22} (v_{21}, v_{31}, \cdots, v_{n1})^T = 0,\\
\vec{h}_{12} u_{11} + (u_{12}, u_{13}, \cdots, u_{1n}) H_{22} = 0,
\end{flalign*}
that is,
\begin{eqnarray}
(v_{21}, v_{31}, \cdots, v_{n1})^T &=& -H_{22}^{-1} \vec{h}_{21} v_{11},\nonumber\\
(u_{12}, u_{13}, \cdots, u_{1n})  &=& -\vec{h}_{12} H_{22}^{-1} u_{11}.
\end{eqnarray}
Since $\sum_{i = 1}^n |u_{1i}|^2 = \sum_{i = 1}^n |v_{i1}|^2 = 1$, then
\begin{eqnarray}
|u_{11}|^2 = \frac{1}{1 + |\vec{h}_{12} H_{22}^{-1}|^2} \leq \frac{|H_{22}|^2}{|H_{22}|^2 + |\vec{h}_{12}|^2},\nonumber\\
|v_{11}|^2 = \frac{1}{1 + |H_{22}^{-1} \vec{h}_{21}|^2} \leq \frac{|H_{22}|^2}{|H_{22}|^2 + |\vec{h}_{21}|^2},
\end{eqnarray}
where `$=$'s are achieved only when every column of $H_{22}$ parallels with $(u_{12}, \cdots, u_{1n})^\dagger$ and every row of $H_{22}$ parallels with $(v_{21}, \cdots, u_{n1})^\ast$.\\
Thus,
\begin{eqnarray}
P_I &=& P_I(1,1) = |q_{11}|^2\nonumber\\
&=& |u_{11} \vec{h}_{12} H_{22}^{-1} \vec{h}_{21} v_{11}|^2\nonumber\\
&=& \frac{|\vec{h}_{12} H_{22}^{-1} \vec{h}_{21}|^2}{(1 + |\vec{h}_{12} H_{22}^{-1}|^2)(1 + |H_{22}^{-1} \vec{h}_{21}|^2)}\nonumber\\
&\leq & \frac{|\vec{h}_{12}|^2 |H_{22}|^2 |\vec{h}_{21}|^2}{(|H_{22}|^2 + |\vec{h}_{12}|^2)(|H_{22}|^2 + |\vec{h}_{21}|^2)}.
\end{eqnarray}
Since $|\vec{h}_{12}|^2 + |H_{22}|^2 + |\vec{h}_{21}|^2 = 1$, we obtain
\begin{eqnarray}
\max P_I =\max |q_{11}|^2 = \frac{5\sqrt{5} - 11}{2}.
\end{eqnarray}
$P_I$ reaches this maximum only when $|\psi\rangle=a |01\rangle + e^{i \theta}\sqrt{1-2a^2} |10\rangle + a |11\rangle$, up to some local unitary transformations, where $a = \sqrt{(3-\sqrt{5})/2}$ and $\theta$ is an arbitrary angle.
\end{proof}

\section{The Second type of extensions of Hardy's Paradox}
The second type \cite{JAZ13} of extensions of the original Hardy's paradox into high-dimension reads
\begin{eqnarray}
P(B_1 < A_1) = P(A_1 < B_2) = P(A_2 < B_1) = 0,\nonumber\\
P_{II} = P(A_2 < B_2) > 0,
\end{eqnarray}
which is equivalent to the violation of a tight Bell inequality---the CGLMP inequality \cite{CGLMP02,ZG08}.

Since $P(B_1 < A_1) = P(A_1 < B_2) = P(A_2 < B_1) = 0$, then
\begin{eqnarray}\label{cond}
h_{ij} = 0\ \text{if}\ i > j,\nonumber \\
h'_{ij} = 0\ \text{if}\ i < j,\nonumber\\
h''_{ij} = 0\ \text{if}\ i < j.
\end{eqnarray}
This determines that
\begin{eqnarray}\label{setting}
U^\ast &=& F. \text{Orthogonalize}[F. H^T],\nonumber\\
V^\dagger &=& \text{Orthogonalize}[H],
\end{eqnarray}
where $F$ is a constant matrix with $f_{i, n+1-i} = 1$ and the rest being all zero, and `Orthogonalize' means the Gram-Schmidt process.\\

Although $\max P_{II}$ increases with the dimension as shown in ref.\cite{JAZ13}, here will prove that $\max P_{II}(1,2)$ is independent of dimension.
\begin{theorem}
$
\max P_{II}(1,2) =\max |q_{12}|^2 = \frac{5\sqrt{5} - 11}{2}.
$
\end{theorem}

\begin{proof}
Without lose of generality, we similarly assume $H_{22}$ is invertible here and $P_{II}(1,2)$ is continuous.\\
By $Q = U H V$ and conditions\eqref{cond}, we know $q_{12} = u_{11} h_{11} v_{12}$ and
\begin{eqnarray}
\vec{h}_{12} u_{11} + (u_{12}, u_{13}, \cdots, u_{1n}) H_{22} = 0,
\end{eqnarray}
thus,
\begin{eqnarray}
(u_{12}, \cdots, u_{1n}) = - \vec{h}_{12} H_{22}^{-1} u_{11}.
\end{eqnarray}
Since $\sum_{i = 1}^n |u_{1i}|^2 = 1$, then
\begin{eqnarray}
|u_{11}|^2 = \frac{1}{1 + |\vec{h}_{12} H_{22}^{-1}|^2} \leq \frac{|H_{22}|^2}{|H_{22}|^2 + |\vec{h}_{12}|^2},
\end{eqnarray}

Through the Gram-Schmidt process and the fact that $v_{12}^\ast$ equals the element in the second row and first column of $V^\dagger$, we know
\begin{eqnarray}
|v_{12}|^2 &=& |v_{12}^\ast|^2\nonumber\\
&=& \frac{|h_{11}|^2 |\langle \vec{h}_{12}, \vec{t} \rangle|^2}{(|h_{11}|^2 + |\vec{h}_{12}|^2)^2|\vec{t}|^2  - (|h_{11}|^2 + |\vec{h}_{12}|^2) |\langle \vec{h}_{12}, \vec{t} \rangle|^2}\nonumber\\
&\leq& \frac{|\vec{h}_{12}|^2}{|h_{11}|^2 + |\vec{h}_{12}|^2},
\end{eqnarray}
where `$=$' is achieved only when $|\langle \vec{h}_{12}, \vec{t} \rangle| = |\vec{h}_{12}||\vec{t}|$.\\
Thus
\begin{eqnarray}
|q_{12}|^2 \leq \frac{|h_{11}|^2 |\vec{h}_{12}|^2 |H_{22}|^2}{(|h_{11}|^2 + |\vec{h}_{12}|^2)(|H_{22}|^2 + |\vec{h}_{12}|^2)},
\end{eqnarray}
since $|h_{11}|^2 + |\vec{h}_{12}|^2 + |H_{22}|^2 \leq 1$, where `$=$' is achieved only when $\vec{h}_{21} = 0$. As a result, we finally obtain
\begin{eqnarray}
\max P_{II}(1,2) =\max |q_{12}|^2 = \frac{5\sqrt{5} - 11}{2}.
\end{eqnarray}
\end{proof}

\section{A Conjecture}

The fact that $\sum_{i<j}^n P_{II}(i,j)$ for $n=3, 4, 5, 6$ are also independent of the system's dimension leads us to a conjecture:
\begin{theorem}
For an arbitrary $n$, $\max \sum_{i<j}^n P_{II}(i,j)$ is independent of the system's dimension $k$, i.e., device-independent. The optimal state is equivalent to the standard form
\[
\begin{bmatrix}
H_n &  \\
    & \mathbf{0}_{k-n}\\
\end{bmatrix},
\]
where $H_n$ is a optimal state for $\max \sum_{i<j}^n P_{II}(i,j)$ in the $n$-dimension system.
\end{theorem}

Here we list some numerical results: $\sum_{i<j}^3 P_{II}(i,j) = 0.141327$ for $3, 4, 5$-dimension system,
\begin{eqnarray}
H_3 = \left(
\begin{array}{ccc}
 0.498328 & 0.316483 & 0.329301 \\
 0 & 0.441108 & 0.316483 \\
 0 & 0 & 0.498328 \\
\end{array}
\right)
\end{eqnarray}
which is the optimal state for three-dimensional Hardy's paradox of the second type, and
\begin{eqnarray}
H_{3,4} &=& \left(
\begin{array}{cccc}
 0.498328 & 0.316483 & 0.329301 & 0\\
 0 & 0.441108 & 0.316483 & 0\\
 0 & 0 & 0.498328 & 0\\
 0 & 0 & 0 & 0\\
\end{array}
\right)\nonumber\\
&=&\begin{bmatrix}
H_3 &  \\
    & \mathbf{0}_1\\
\end{bmatrix}
\end{eqnarray}
which is an optimal state in the standard form.
\begin{eqnarray}
H_{3,5} =\hspace{20em} \nonumber\\
\left(
\begin{array}{ccccc}
 0.49832 & 0.316487 & 0.232321 & 0.187338 & 0.139177 \\
 0 & 0.441109 & 0.223283 & 0.18005 & 0.133762 \\
 0 & 0 & 0.351577 & 0.283503 & 0.210619 \\
 0 & 0 & 0 & 0 & 0 \\
 0 & 0 & 0 & 0 & 0 \\
\end{array}
\right)
\end{eqnarray}
which is also an optimal state. If we redefine the third element of $B_2$'s basis as $(0, 0, 0.351577, 0.283503, 0.210619)$, then $H_{3,5}$ changes to the standard form
\[
\left(
\begin{array}{ccccc}
 0.498318 & 0.316486 & 0.329299 & 0. & 0. \\
 0. & 0.441107 & 0.316488 & 0. & 0. \\
 0. & 0. & 0.498336 & 0. & 0. \\
 0. & 0. & 0. & 0. & 0. \\
 0. & 0. & 0. & 0. & 0. \\
\end{array}
\right)
\approx
\begin{bmatrix}
H_3 &  \\
    & \mathbf{0}_2\\
\end{bmatrix}.
\]
That is, $H_{3,5}$ is equivalent to the standard form.

For $H_n, n = 2, 3, \cdots, 7$ and the corresponding maximum values see the ref. \cite{JAZ13} for details.

The merits of such a conjecture are twofold. Theoretically, it implies that only a few terms in $P_{II}$ play an active role in revealing the nonlocal
feature of quantum states and this specific set remains the same for every dimension $n$. In experiments, it allows one to measure fewer joint correlations, with lower measuring deviations and less effort in state preparation, to obtain a greater quantum violation of the local realistic theory.

Moreover, we further conjecture that, while the first type of extension provides a minimal generalization \cite{KC05} of Hardy's paradox, the second type of extension maybe provides a maximal one.

\section{Conclusions}

We have reproved the dimension-independency of the first type of extensions of Hardy's paradox, and demonstrated that if only one particular term $P_{II}(1,2)$ were considered, the maximum probability of nonlocal events in the second type could be also dimension-independent. With numerical evidence, we further conjectured that if more than one particular terms $\sum_{i<j}^n P_{II}(i,j)$ were considered, then the maximum probability of nonlocal events is also independent of dimensions. Hence, the role of dimension played in Hardy's paradox is this: it determines which term or how many terms can be considered in the paradox. We also note that while the first type of extensions provides a minimal generalization \cite{KC05} of Hardy's paradox, the second type of generalization may provide a maximal one.  Our results may stimulate studies on a promising unified perspective of Hardy's paradox in further investigations.


\begin{acknowledgments}
J.L.C. is supported by the National Basic Research Program (973
Program) of China under Grant No.\ 2012CB921900 and the NSF of China
(Grant Nos. 11175089 and 11475089). This
work is also partly supported by the National Research Foundation
and the Ministry of Education, Singapore.
\end{acknowledgments}


\begin{thebibliography}{99}

\bibitem{Hardy92}
 L. Hardy,
 \href{http://prl.aps.org/abstract/PRL/v68/i20/p2981_1}{Phys. Rev. Lett. \textbf{68}, 2981 (1992).}

\bibitem{Hardy93}
 L. Hardy,
 \href{http://dx.doi.org/10.1103/PhysRevLett.71.1665}{Phys. Rev. Lett. \textbf{71}, 1665 (1993).}

\bibitem{Goldstein94}
 S. Goldstein,
 \href{http://dx.doi.org/10.1103/PhysRevLett.72.1951}{Phys. Rev. Lett. \textbf{72}, 1951 (1994).}

\bibitem{cereceda04}
 J. Cereceda, \href{http://www.sciencedirect.com/science/article/pii/S0375960104007777}{Phys. Lett. A 327, 433 (2004).}

\bibitem{CBTB13}
 A. Cabello, P. Badzi\c{a}g, M. Terra Cunha, and M. Bourennane,
 \href{http://prl.aps.org/abstract/PRL/v111/i18/e180404}{Phys. Rev. Lett. \textbf{111}, 180404 (2013).}

\bibitem{Mermin95}
 N. D. Mermin,
 in {\em Fundamental Problems in Quantum Theory},
 edited by D. M. Greenberger and A. Zeilinger,
 \href{http://onlinelibrary.wiley.com/doi/10.1111/j.1749-6632.1995.tb39001.x/abstract}{Ann. N. Y. Acad. Sci. \textbf{755}, 616 (1995).}


\bibitem{KC05}
 S. Kunkri and S. K. Choudhary,
 \href{http://dx.doi.org/10.1103/PhysRevA.72.022348}{Phys. Rev. A \textbf{72}, 022348 (2005).}

\bibitem{SG11}
 K. P. Seshadreesan and S. Ghosh,
 \href{http://dx.doi.org/10.1088/1751-8113/44/31/315305}{J. Phys. A: Math. Theor. \textbf{44}, 315305 (2011).}

\bibitem{RZS12}
 R. Rabelo, L. Y. Zhi, and V. Scarani,
 \href{http://prl.aps.org/abstract/PRL/v109/i18/e180401}{Phys. Rev. Lett. \textbf{109}, 180401 (2012).}

\bibitem{JAZ13}
J. L. Chen, A. Cabello, Z. P. Xu, H. Y. Su, C. F. Wu, and L. C. Kwek, \href{http://journals.aps.org/pra/abstract/10.1103/PhysRevA.88.062116}{Phys. Rev. A \textbf{88}, 062116 (2013)}.


\bibitem{Hardy97}
 L. Hardy,
 in {\em New Developments on Fundamental Problems in Quantum Physics},
 edited by M. Ferrero and A. van der Merwe
 (Kluwer, Dordrecht, Holland, 1997), p.~163.

\bibitem{BBDH97}
 D. Boschi, S. Branca, F. De Martini, and L. Hardy,
 \href{http://prl.aps.org/abstract/PRL/v79/i15/p2755_1}{Phys. Rev. Lett. \textbf{79}, 2755 (1997).}

\bibitem{CGLMP02}
 D. Collins, N. Gisin, N. Linden, S. Massar, and S. Popescu,
 \href{http://prl.aps.org/abstract/PRL/v88/i4/e040404}{Phys. Rev. Lett. \textbf{88}, 040404 (2002).}

\bibitem{ZG08}
 S. Zohren and R. D. Gill,
 \href{http://dx.doi.org/10.1103/PhysRevLett.100.120406}{Phys. Rev. Lett. \textbf{100}, 120406 (2008).}
\end{thebibliography}
\end{document}